\documentclass[onecolumn,draftcls,conference]{ieeetran}
\usepackage{amsmath, amsfonts, amssymb}

\usepackage{optidef}
\usepackage{color}
\usepackage{bm}
\usepackage{tikz,pgfplots}
\usepackage{algorithm}
\usepackage{cite}
\usepackage{algpseudocode}
\usepackage{siunitx}
\usepackage{subcaption}

\usepackage{accents}
\newcommand{\dbtilde}[1]{\accentset{\approx}{#1}}

\newcommand{\iif}[1] {\ \mathbf{1} {\big({#1}\big)}  }
\newcommand{\eexp}[1] {\mathbb{E}\left\{{#1}\right\}}

\DeclareMathOperator*{\argmax}{arg\,max}
\DeclareMathOperator*{\argmin}{arg\,min}

\def \oot {{1 \over 2}}
\def \C {\mathcal{C}}
\def \Ut {U_0}
\def \Am {A_{\text{m}}}

\def \Pam {P_{ \text{APm}  }}

\def \No {N_o}
\def \Ni {N_i}
\def \Cm {C_{\text{m}}}
\def \Pm {P_{\text{m}}}

\def \Out {\mathcal{O}}
\def \In {\mathcal{I}}

\newtheorem{lemma}{Lemma}
\newtheorem{theorem}{Theorem}
\newtheorem{remark}{Remark}
\usetikzlibrary{arrows}
\usetikzlibrary{arrows,decorations.pathmorphing,backgrounds,positioning,fit,petri,shapes.geometric}
\title{Joint Data Routing and Power Scheduling for Wireless Powered Communication Networks }
\author{Mohammad Movahednasab, Naeimeh Omidvar, Mohammad Reza Pakravan, Tommy Svensson}
\begin{document}
\maketitle
\begin{abstract}
 In a wireless powered communication network (WPCN), an energy access point supplies the energy needs of the network nodes through radio frequency wave transmission, and  the nodes store the received energy in their batteries for their future data transmission. In this paper, we propose an online stochastic policy that jointly controls energy transmission from the EAP to the nodes and data transfer among the nodes. For this purpose, we first introduce a novel perturbed Lyapunov function to address the limitations on the energy consumption of the nodes imposed by their batteries.  Then, using Lyapunov optimization method, we propose a policy which is adaptive to any arbitrary channel statistics in the network. Finally, we provide theoretical analysis for the performance of the proposed policy and show that it stabilizes the network, and the average power consumption of the network under this policy is within a bounded gap of the minimum power level required for stabilizing the network.

\end{abstract}

\section{Introduction}
Nowadays, smart electronic devices are increasingly making their way into our daily life. It is predicted that by 2021, there will be around 28 billion connected devices  all over the world \cite{Ab2016}, a great number of which will be portable and battery-powered. However, in some applications, replacing the batteries or recharging them by cables is impossible, e.g. in biomedical implants inside human bodies \cite{Zeng2016} or distributed monitoring  sensors in a wide area of forest. Consequently, to ensure a better user experience for the next-generation networks, the problem of providing the required power for the portable battery-operated devices has recently gained lots of attention, both from academia and industry \cite{Huang2015,Zeng2016}. Recently, the idea of charging batteries over the air is considered as a solution which guarantees an uninterrupted connection and operates autonomously, while reduces the massive battery disposal.   Wireless Power Transfer (WPT) is the key enabling technology for charging over the air. There are various WPT methods including \textit{Radio Frequency} (RF) power transfer\cite{Huang2015}, \textit{resonant coupling}\cite{Resonant2009} and \textit{inductive coupling}\cite{Inductive2013}. Compared to the two latter methods, RF power transfer provides a wider coverage range and is more flexible for transmitter/receiver deployment and movement  \cite{Zeng2016}.  Therefore, it is considered as the most promising WPT approach by the literature.

Adapting  WPT technology in wireless communication networks introduces new research challenges, mostly related to  increasing coverage and efficiency. A prominent challenge is how to maintain power transfer efficiency despite the transmission path loss \cite{Yang2014}. There has been numerous studies on energy beamforming as a technique for alleviating the high transmission path loss (e.g., see \cite{Yang2014,Liu2014,Larsson2013,Nasir2016}). In \cite{Yang2014} and  \cite{Liu2014}, a wireless powered communication network (WPCN) consisting of a hybrid data/energy access point with multiple antennas and several single-antenna users is considered, in which the  access point transmits energy toward the nodes in the down-link direction and the nodes transmit data to the access point in the up-link direction. In these works, the minimum achievable rate among users is maximized by optimizing the beamforming vector and some other controllable parameters. 
Energy beamforming for the so-called simultaneous wireless information and power transfer (SWIPT) method is studied in \cite{Larsson2013,Nasir2016}. Under SWIPT, both energy and data are  jointly transmitted by an RF carrier in the down-link, the receiver extracts data or harvests power through splitting the received signal in the time or power domain.

Cooperative wireless powered communication is another line of research that aims at increasing the network coverage (e.g. see \cite{Ju2014,Bafghi2017,Xu2017,Gurakan2016}). The intuition behind it is that in WPCNs, the  users nearer  to the access point harvest more energy, while need to consume less energy for  their data transmission. Hence, using cooperation, these nodes can use some of their surplus energy to help relaying the data of the further nodes. In \cite{Ju2014}, a two user scenario is considered in which the nearer user allocates a portion of its harvested energy to help relaying the farther user's data to the access point. The authors have  maximized the sum rate of the users by optimizing the resource allocation. In \cite{Bafghi2017}, the authors have derived the achievable rate of a two-hop relay network, in which the relay stores its harvested energy in the battery for its future transmissions. While most works in the  related literature have considered two node cooperation there are very few works on multi-hop cooperation, e.g.  \cite{Xu2017} and \cite{Gurakan2016}. In these works  a general multi-hop network with energy transfer capability has been considered, and the  routing policy and energy allocation are determined so as to maximize the sum rate and the lifetime  of the network. 

It should be noted that most of the existing works in the literature  have focused on optimizing the network parameters for a single time-slot. Clearly, this approach is not optimal when the users can store the harvested energy in their batteries for their future use. There are very few works on the long-term network optimization \cite{Biason2018,Choi2015, RezaeiPIMRC,RezaeiGlobe,RezaeiArxiv}. The authors in \cite{Biason2018} have studied the long-term network utility optimization through markov decision programming (MDP) theory. The MDP method requires  statistical knowledge of the channel variation, and the complexity of its solutions grows fast as the network dimension increases \cite{Bertsekas2005}. However, the Lyapunov optimization technique applied in  \cite{Choi2015,RezaeiPIMRC,RezaeiGlobe,RezaeiArxiv}, is independent of the channel statistics and the network dimension. 

In this work, we consider the problem of designing an optimal WPT policy that schedules power allocation, data routing and energy beamforming in a multi-hop WPCN.  We consider a battery level  constraint for each node, which indicates that at each time-slot, the energy that can be consumed can not be greater than the  energy stored in the battery. Moreover, the average data backlog in the network should remain finite. The battery constraint complicates the problem, since high energy consumption at a time may highly lower down the battery level and cause energy outage in future.  Therefore, the decision at one  time-slot affects the optimal decision in future, as well. This coupling makes finding the optimal policy highly challenging. A similar problem has been considered in utility maximization for  energy harvesting wireless sensor networks in \cite{Huang2013,Gatzianas2010}. The authors have addressed the battery constraint by a modified Lyapunov optimization method. However, their method is not applicable in our energy optimization problem, as their objective function (and hence, their accordingly  analysis) is totally different. In this paper, we use  Lyapunov optimization method with a novel Lyapunov function to avoid energy outage. We propose an online policy that is independent of the channel statistics. Under this policy, at each time-slot, the energy beam is focused toward the nodes with lower battery levels, greater queue backlogs and better energy link condition. Moreover, the data is routed through the  nodes with less congested queues and greater battery level. We then analyze the performance of the proposed policy and provide theoretical results that show the performance of our policy is within $\mathcal{O}({1\over V})$ of the optimal policy, for any $V>0$, while the average backlogs of data queues are upper bounded by $\mathcal{O}(V)$. We would like to note that the most related work to our paper is  \cite{Choi2015}, which studies energy optimization  in a single-hop WPCN. However, the authors in  \cite{Choi2015} have pursued a different approach to address the energy outage problem, which imposes a minimum requirement on maximum transmission power of the access point. This minimum requirement can be too high in practice and may not be satisfied in certain cases, due to safety or implementation issues. 

The contributions of this paper can be summarized as follows: 
\begin{itemize}
\item We propose a  power scheduling, energy beamforing and data routing policy for a general multi-hop WPCN. 
\item We show that our policy conforms to the battery level constraint.
\item  Using Lyapunov optimization method, we bound the optimality gap of the EAP average power consumption and the average backlog of the queues. 
 \end{itemize}
The rest of the paper is organized as follows.  Section \ref{sec:sysMo} illustrates  our system model and problem formulation.  Section \ref{sec:policy}  presents our proposed policy.    The performance of the policy is analyzed in Section \ref{sec:performance}. Simulation results are presented in Section \ref{sec:sim}, and finally, Section \ref{sec:conclude} concludes the paper.

Notation: We use boldface letters to denote matrices and vectors. $(.)^T$  denotes the transpose of a matrix. $|.|$ denotes the absolute value. If not mentioned, vectors are single row-matrices. $\mathbb{E}\{.\}$ represents the expectation. $[x]^{+}$ denotes $\max\{x,0\}$. $\iif{Condition}$ equals $1$ if the $Condition$ is satisfied and equals $0$, othewise.

\section{System Model}\label{sec:sysMo}

We consider a WPCN consisting of one energy access point (EAP) and $N$ wireless nodes, where there exist $S$ streams of data between distinct endpoints in the network.  The nodes are battery-powered, and the batteries are recharged by the energy received from the EAP.  
%The links between EAP and WNs are only for energy transmission so they are marked by energy links in Fig. \ref{fig:sampleTopology} and the links between WNs are for data transmission and are marked by data links.
 There exist $N$ \textit{energy links} between the EAP and the nodes, and $L$  \textit{data links} between the nodes. The topology of a sample network  is depicted in Fig. \ref{fig:sampleTopology}. For each data link $l\in\{1, \ldots, L\}$, $T(l)$ and $R(l)$ denote the transmitter and receiver of  the $l$-th link, respectively. Moreover, we define $\mathcal{I}_n$ and $\mathcal{O}_n$ to be  the sets of the ingoing and outgoing data links of node $n$, respectively.  The time horizon is divided into time-slots with fixed length\footnote{Without loss of generality, we assume the slot duration is normalized to 1. Therefore, we sometimes use the terms ``power'' and ``energy'' interchangeably.}, indexed by $t$. At the beginning of each time-slot, a small portion of it is devoted to channel estimation and control signaling. The rest of the time-slot is divided equally for energy and data transmission, respectively. The EAP is equipped with $M$ antennas to focus its transmission beam toward the nodes. Moreover, we assume that the nodes use a single antenna for both energy reception and data transmission/reception. 

The channels state information are assumed to be constant during a time-slot but vary randomly and independently in successive time-slots.  At each time-slot $t$, $g_l(t)$ and $h_n^m(t)$ represent the channel gain of the $l$-th data link and the gain of the channel between $m$-th antenna of the EAP and the node $n$, respectively. Accordingly, we define  $\bm{g}(t) \triangleq (g_1(t),\ldots, g_L(t))$ and $\bm{h}_n(t) \triangleq (h_n^1(t), \ldots, h_n^M(t))$ as the channel vectors for data links and energy link of node $n$, respectively.

\begin{figure}
\centering
\includegraphics[width = 0.45 \textwidth]{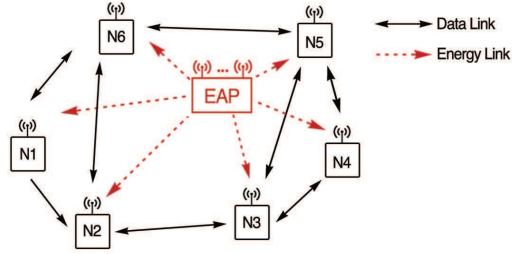}
\caption{A sample  WPCN topology.}
\label{fig:sampleTopology}
\end{figure}

\subsection{Data and Energy Transmission}
Let $\bm{p}(t) \triangleq (p_1(t), \ldots,p_L(t))$ denote the data link \textit{power vector}, in which the $l$-th entry  specifies the transmission power over the $l$-th data link. 
Moreover, let $\Pi$ denote the set of all feasible power vectors. We assume that setting any element of a power vector in $\Pi$ to zero results in a new power vector that also belongs to  $\Pi$. Furthermore, we assume that the peak transmission power is limited to $\Pm$  (i.e., $p_l(t) \in [0,\Pm]$). Let $C_l(\bm{p}(t), \bm{g}(t))$ denote the data transmission capacity of link $l$ under power vector $\bm{p}(t)$ and channel vector $\bm{g}(t)$. Some important properties of   $C_l(\bm{p}(t), \bm{g}(t))$  is presented in following remark.

%Let us denote by $\bm{p}^\ast_0(t)$, a power vector with all link powers equal to $\bm{p}^\ast(t)$ but  the power in the $l$-th link, which is set to zero.

\begin{remark} \label{rem:capProp}
Consider two power vectors $\bm{p}(t)$ and $\bm{p}^{\prime}(t)$, where $p_{l^\prime}^{\prime}(t) = p_{l^\prime}(t),\forall l^\prime\neq l$ and $p_{l}^{\prime}(t) = 0$. The capacity of link $l$ under each of these two power vectors satisfies the following properties:
 \begin{align}
 C_l(\bm{p}^\prime(t), \bm{g}(t)) &= 0,\label{eq:noPowerNoRate}\\
 C_l(\bm{p}(t), \bm{g}(t))& \leq \delta p_l(t),\label{eq:rateBound}\\
C_{l^\prime}(\bm{p}(t), \bm{g}(t)) &\leq C_{l^\prime}(\bm{p}^\prime(t), \bm{g}(t)) \;\; \forall {l^\prime} \neq l,\label{eq:interference}
\end{align}
where \eqref{eq:rateBound} holds for some $\delta >0$.
\end{remark}
Note that the above properties are satisfied under conventional rate-power functions. Equation \eqref{eq:noPowerNoRate} implies that for any link $l$, no data can be passed through it if no power is assigned to this link.  Inequality \eqref{eq:rateBound} states that the rate-power function is upper bounded by a linear function, which is the case for differentiable functions with limited first derivative. Finally, inequality \eqref{eq:interference} holds due to the interference effect among wireless links.  Furthermore, we assume that there exists a constant $\Cm > 0$, such that $C(\bm{p}(t), \bm{g}(t))\leq \Cm, \;\forall \bm{p}(t) \in \Pi,\forall\bm{g}$. For any link $l$, let $C_l^s(t)$ denote the transmission rate allocated to stream $s, \; \forall s$, over that link. Clearly, the sum allocated rate over each link $l$ should not exceed the capacity of that link. Therefore, a feasible rate allocation scheme should satisfy
\begin{align}
\sum_{s=1}^SC_l^s(t) \leq C_l(\bm{p}(t), \bm{g}(t)).
\label{eq:capLimit}
\end{align}

Fig. \ref{fig:EAPStructure} shows the considered structure of an EAP. The EAP performs energy beamforming  to concentrate its transmit energy towards the nodes.  Vector 
$\bm{w}(t) \triangleq (w_1(t), \ldots,w_M(t))$ denotes the normalized beamforming vector of the EAP, and accordingly, the received power at each node $n$ is  given by

\begin{align}
Q_n(t) = p_{AP}(t) |\bm{w}(t)\bm{h}_n^T(t)|^2\;\; \forall n,
\label{eq:recPower}
\end{align}
where $p_{AP}(t)$ is the EAP's transmit power at time $t$, with its peak power equal to $\Pam$, i.e., $p_{AP}(t) \in [0,\Pam]$.

\begin{figure}

\centering
\includegraphics[width = 0.3 \textwidth]{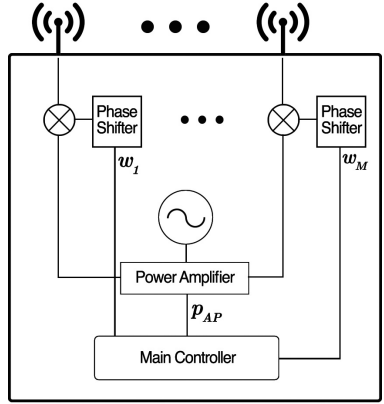}
\caption{  Internal structure of  an EAP.}
\label{fig:EAPStructure}
\end{figure}

%\begin{itemize}
%\item Consider a specific power vector $\bm{p}^\ast(t)$, with nonzero element at $l$-th entry ($p^\ast_l(t) > 0$) and construct  $\bm{p}^\ast_0(t)$, by setting the $l$-th entry of  $\bm{p}^\ast(t)$ to zero, then we have,
% \begin{align}
% C_l(\bm{p}^\ast_0(t), \bm{g}(t)) = 0
%\end{align}
%and
% \begin{align}
% C_l(\bm{p}^\ast(t), \bm{g}(t)) \leq \delta p^\ast_l(t)
%\end{align}
%\item For the 
% \begin{align}
%C_{l^\prime}(\bm{p}^\ast(t), \bm{g}(t)) \leq C_{l^\prime}(\bm{p}_0^\ast(t), \bm{g}(t)) \;\; \forall {l^\prime} \neq l.
%\end{align}
%\end{itemize}
%
%
%As a general model we 
%
%In the most general model, we assume the channel capacity is a function of 
% Considering a  general wireless network with possibly interfering links, the transmission rate on any link will be a function of  $\bm{p}(t) \triangleq (p_1(t), \ldots,p_L(t))$ which is the vector of transmission powers allocated to data  links and $\bm{g}(t) \triangleq (g_1(t),\ldots, g_L(t))$. We denote by $C_l(\bm{g}(t),\bm{p}(t))$ the transmission capacity on link $l$ during time slot $t$.

\subsection{Wireless Nodes }
As shown in Fig. \ref{fig:WNStructure}, each node includes  $S$ data queues and  is equipped with a battery.  Let $U_n^s(t)$ denote the backlog of the data queue allocated to stream $s \in \{1,\ldots,S\}$ in node $n$. The backlog evolves as follows
\begin{align}
\begin{split}
U_n^s(t+1) = \bigg [ U_n^s(t) - \sum_{l \in \mathcal{O}_n} C^s_l(t)\bigg]^{+}+\sum_{l \in \mathcal{I}_n} C^s_l(t) +A^s_n(t),
\end{split}
\label{eq:queueEvolve}
\end{align}
where $A_n^s(t) \in [0,\Am]$ denotes  the random arrival process of data stream $s$ at node $n$. Note that $A_n^s(t)$ is nonzero only if node $n$ is the source of stream $s$.  Let $\lambda^s_n \triangleq \eexp{A_n^s(t)}$, then clearly we have
\begin{align*}
\lambda_n^s = 
\begin{cases}
\lambda_s, \;\;&\text{if node $n$ is source of stream $s$, }\\
0, \;\;&\text{otherwise,}
\end{cases}
\end{align*} 
where $\lambda_s$ is the arrival rate for stream $s$. 

The battery of each node  is recharged by the energy received from the EAP and is discharged when the node transmits data. Let $E_n(t)$ denote the battery level of node $n$ at beginning of time-slot $t$. Therefore, the battery level at node $n$ evolves according to following equation:
\begin{align}
E_n(t+1) = E_n(t) - \sum_{l \in \Out_n} p_l(t) + Q_n(t),
\label{eq:batteryEvolve}
\end{align}
where $ \sum_{l \in \Out_n} p_l(t)$ is the total transmit power of node $n$ at time-slot $t$. 
\begin{figure}

\centering
\includegraphics[width = 0.3 \textwidth]{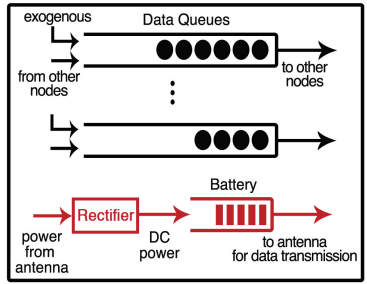}
\caption{ The data flows and energy flow inside each node of the network.}
\label{fig:WNStructure}
\end{figure}

\subsection{Network Controller} 

 There exists a  network controller, located at the EAP that controls over both data and energy links, having access to channel state information, data queue backlogs and the battery levels of all links and nodes in the network. It controls the energy links by specifying the EAP transmission power $p_{AP}(t)$  and the beamforming vector $\bm{w}(t)$, and the data links  by determining their power vector $\bm{p}(t)$ and   routing of the data streams. 
It routes the data through allocating the capacity of the links  to the existing data streams in the network. 

%
%At the beginning of each time-slot the controller indicates 
%
%
%or a separate node, needs to know the channel state information, the backlog of data queues and the battery status of the nodes to determine the energy and data transmission policy in current time slot. This policy is signaled to the nodes in the first sub-slot. The details of channel estimation method and signaling is out of scope of this paper since we are focused on designing the controller. In the second sub-slot the EAP radiates power toward WNs and the WNs recharge their batteries by the received power. Finally in the third sub-slot data is communicated between nodes.

%\begin{figure}
%\centering
%\begin{tikzpicture}
%\draw (-4.5,0) [fill=red!20] rectangle node{Channel Est. and Control Sig.} (0,1);
%\draw (0,0)   rectangle node{Power Transmission} (4,1);
%\draw (4,0)  rectangle node{Data Transmission} (8,1);
%\draw[<->,thick] (-4.5,-0.5) -- (8,-0.5) node[midway,below] {Slot Duration};
%\end{tikzpicture}
%\caption{Time Slots Structure}
%\label{fig:time}
%\end{figure}
%The network controller functions according to a \textit{policy}. 
As aforementioned, in this paper, we focus on designing  a joint data routing and energy transfer control policy   for the network controller that minimizes the total transferred power of the EAP while guaranteeing all data queues in the network to be stable.  The intended policy can be explicitly formulated as the solution to the following problem\footnote{Note that including both time average and  expectation in this problem formulation is due to the fact that we seek the optimal policy among both stationary and non-stationary policies. For a non-stationary policy, the expectations are time-dependent and hence, taking the time average is required. }. 
%(in terms of $\bm{p}(t),\bm{w}(t),p_{AP}(t), C_l^s(t) \;\; \forall s,t$),
%
% prevents backlog in data queue to grow infinitely
%
%In this paper we aim to minimize the average power consumption of the network, while the network is stable. We seek an optimal stochastic policy for the network controller that adapts itself to channel state distribution.
%
%
%for The optimal policy can be formulated as the solution to following problem, 

\begin{mini!}
{\bm{p}(t),\bm{w}(t),p_{AP}(t), C_l^s(t) } {\bar{p}_{AP} =  \lim_{T \rightarrow \infty}\frac{1}{T}\sum_{t=0}^{T-1}\mathbb{E}\left\{p_{AP}(t)\right\}\label{eq:objectiveFun}}{\label{prob:mainProbDef}}{}
\addConstraint{\sum_{l \in \mathcal{O}(n)}p_l(t) \leq E_{n}(t), \;\;\forall n,t \label{eq:batteryConstraint}}
\addConstraint{\bar{U} = \limsup_{T \rightarrow \infty}{1 \over T}\sum_{t = 0}^{T-1}\sum_{n,s}\mathbb{E}\left\{U_n^s(t)\right\} < \infty \;\; \forall n,s \label{eq:stableConstraint}}
\addConstraint{\eqref{eq:capLimit},\eqref{eq:queueEvolve}, \eqref{eq:batteryEvolve},}
\end{mini!}
where Constraint \eqref{eq:batteryConstraint} guarantees that the sum power allocated to the outgoing links of a node is not greater than what can be supported by the battery level of the node. Moreover, Constraint \eqref{eq:stableConstraint} ensures stability of all queues.  

Note that the formulated problen in  \eqref{prob:mainProbDef} is a stochastic utility optimization problem.   Although in general, these problems can be tackled by the so-called min drift plus penalty (MDPP) algorithm \cite{Neely2010},  the battery constraint in \eqref{eq:batteryConstraint} highly complicates our problem and makes it quite challenging. This is mainly due to the fact that in the battery-operated case, consuming high power in a specific time-slot may drastically lower  the battery level down and  restrict future transmissions. Therefore, having the battery level constraint, policies  with independent decisions at each time-slot are not optimal any more, which is not acceptable in  the  MDPP problem formulation. In the sequel, we propose a solution to handle the battery constraint in MDPP problem formulation.   
For  convenience, all the notations in the paper and their definitions are presented in Table \ref{table:notation}.
\begin{table}
\caption{Notation Summary}
\begin{center}
\begin{tabular}{  |c|c|}
\hline
Symbol                  				& Meaning   \\ \hline
$U_n^s(t)$		 				& The backlog of data queue allocated to stream $s$ in node $n$ at time-slot $t$.\\\hline
$ E_n(t)$               				& The battery level of node $n$ at time-slot $t$. \\\hline
$N, S, L$                       			& Number of nodes, streams and data links. \\\hline
$ T(l)$ 		               			&The transmitter index for  link $l$. \\\hline
$R(l)$            		  			&The receiver index for  link $l$. \\\hline
$ \Out_n$              				&The set of outgoing  links from  node $n$. \\\hline
$ \In_n$              					&The set of ingoing  links to  node $n$. \\\hline
$ \No$                					& The maximum number of outgoing links from from a specific node. \\\hline
$ \Ni$ 			              			& The maximum number of ingoing links to a specific node. \\\hline
$Q_n(t) $               				& The harvested energy by node $n$ at time-slot $t$.\\\hline
$ \bm{w}(t)$						& The beamforming vector. \\ \hline
$p_{AP}(t)$						& The transmission power of EAP at time-slot $t$.\\ \hline
$A_n^s(t)$               			       & The exogenous data arrival for stream $s$ in node $n$.\\\hline
$\lambda_n^s $					& The mean value of the exogenous data arrival for stream $s$ in node $n$.\\\hline
$\bm{g}(t)$						& The vector of  data link channel states at time-slot $t$. \\ \hline
$\bm{h}(t)$						& The vector of  energy link channel states at time-slot $t$. \\ \hline
$\bm{p}(t)$						& The vector of power allocation to data links at time-slot $t$. \\ \hline
$\Pi$					              & The set of valid power vectors. \\ \hline
$C_l(\bm{p}(t), \bm{g}(t))$			& Capacity of link $l$ at time-slot $t$. \\ \hline
$\Pm, \Pam $	 				& Maximum transmission power of nodes and EAP.\\ \hline
\end{tabular}
\end{center}

\label{table:notation}
\end{table}

\section{The Proposed Online Control Policy for Joint Data Routing and Power Transfer Scheduling}\label{sec:policy}
In this section, we present an online control policy for the network controller. This policy is developed based on the Lypunov optimization method \cite{Neely2010}. We propose a novel perturbed Lyapunov function to push the battery level up. The Lyapunov function is defined as
\begin{align*}
L(t)\triangleq \sum_{n=1}^N\sum_{s=1}^S L_n^s(t),
\end{align*}
where
\begin{align}
\begin{split}
&L_n^s(t)\triangleq \oot \big(U_n^s(t)\big)^2 + \oot \big(U_n^s(t)-\C E_n(t)\big)^2 \iif{U_n^s(t) > \C E_n(t)},\\ \label{eq:LyapunovFun}
&\iif{U_n^s(t) > \C E_n(t)} = \begin{cases} 1,    &U_n^s(t) > \C E_n(t),\\ 
0, &\text{otherwise}.
\end{cases}
\end{split}
\end{align}
Note that constant $\C$ in \eqref{eq:LyapunovFun},  is an energy normalization factor and we set it to $\frac{2\delta}{1- \frac{1}{\alpha}}$, where $\delta$ is the slope of the linear upper bound for the rate-power function in \eqref{eq:rateBound} and $\alpha >1$ is a constant and will be discussed later. We also define the Lypunov drift function,
\begin{align*}
\Delta(L(t)) \triangleq \eexp{L(t+1) - L(t)| \bm{U}(t), \bm{E}(t)},
\end{align*}
where $\bm{U}(t)$ and $\bm{E}(t)$ are the sets of all data queues and batteries in the network, respectively. Next we define the drift-plus-penalty function, as follows

\begin{align}
 \Delta(L(t)) +V \mathbb{E}\{   p_{AP}(t)| \bm{U}(t), \bm{E}(t)\},\label{eq:driftPlusPenalty}
\end{align}
where $V>0$  is a control parameter. The following Lemma establishes an upper bound on the above drift-plus-penalty function. 

\begin{lemma}\label{lem:upperBound}
For the defined drift-plus-penalty function, the following inequality always holds
\begin{small}
\begin{align}
\begin{split}
 &\Delta(L(t)) +V \mathbb{E}\{   p_{AP}(t)| \bm{U}(t), \bm{E}(t)\}  \leq B+\sum_{(n,s)\in \tilde{N}(t)} \eexp{ \sum_{l\in \In_n} C_l^s(t) - \sum_{l\in \Out_n} C_l^s(t) | \bm{U}(t), \bm{E}(t) } U^s_n(t) \\
&+ \sum_{(n,s)\in \dbtilde{N}(t)} \mathbb{E}\bigg\{ \C \sum_{l\in \Out_n}p_l(t) + \sum_{l\in \In_n} C_l^s(t)  -\sum_{l\in \Out_n} C_l^s(t) | \bm{U}(t), \bm{E}(t)\bigg\} \big(U^s_n(t) - \C E_n(t)\big)\\
&+ \mathbb{E}\bigg\{  Vp_{AP}(t) -      \C\sum_{(n,s)\in \dbtilde{N}(t)}Q_n(t)\big(U^s_n(t) - \C E_n(t)\big) | \bm{U}(t), \bm{E}(t) \bigg\}\\
&+ \sum_{(n,s)\in \tilde{N}(t)}\lambda_n^sU_n^s(t) + \sum_{(n,s)\in \dbtilde{N}(t)}\lambda_n^s \big(U_n^s(t) - \C E_n(t)\big),
\end{split}\label{eq:dppUpperBound}
\end{align}
\end{small}
where  $\tilde{N}(t)\triangleq \{(n,s) | U^s_n(t) > U_0\}$, $\dbtilde{N}(t)\triangleq \Big\{(n,s) | U^s_n(t) > \max\{U_0, \C E_n(t) \}\Big\}$  and $U_0 =\Pm( \C +\alpha\delta)$. The constant $B$ is defined in Appendix \ref{sec:proofUpperBound}.
\end{lemma}
\begin{proof}
see Appendix \ref{sec:proofUpperBound}.
\end{proof}
The parameter  $U_0$ and the set $\tilde{N}(t)$ in the above Lemma are the \textit{ the congestion  threshold} and the set of \textit{congested queues}, respectively. Moreover, we call the congested queues in the set  $\dbtilde{N}(t)$  \textit{the critically congested queues}, since the size of their backlog exceeds the normalized battery level of their corresponding node. Our policy tends to decrease a queue backlog only if the queue is congested. Consequently,  setting the congestion threshold to the smallest possible value  reduces the average queue backlog.  The parameter $\alpha$ can be optimized to achieve the minimum congestion threshold.  Substituting $\C$ in the definition of $U_0$ with $\frac{2\delta}{1- \frac{1}{\alpha}}$,  it can be verified that $U_0$ is minimized at  $\alpha = \sqrt{2}+1$.

% and choosing an $\alpha$ that minimizes $U_0$ is of interest. Substituting $\C$ with its definition,   it can be verified that $U_0$ is minimized at  $\alpha = \sqrt{2}+1$.

As  will be shown in section \ref{sec:performance},  any policy that minimizes the right hand side of  \eqref{eq:dppUpperBound} at each time-slot  stabilizes the network and yields an average power consumption  within a bounded gap to the optimal power consumption. Consequently, we are interested in finding a policy that minimizes the upper bound in  \eqref{eq:dppUpperBound}. For this purpose, we first rearrange the right hand side in   \eqref{eq:dppUpperBound} and rewrite it as follows:

\begin{align}
\begin{split}
 \Delta(L(t)) +V \eexp{  p_{AP}(t)| \bm{U}(t), \bm{E}(t) }  \leq B &+ \eexp{ \sum_{l=1}^L J_{T(l)}(t) p_l(t)- \sum_{l=1}^L \sum_{s=1}^S W_l^s(t)C_l^s(t) |\bm{U}(t), \bm{E}(t)}\\
&+ \eexp{Vp_{AP}(t) - \sum_{n=1}^N Q_{n}J_n(t)    }\\
&+ \sum_{(n,s)\in \tilde{N}(t)}\lambda_n^sU_n^s(t) + \sum_{(n,s)\in \dbtilde{N}(t)}\lambda_n^s \big(U_n^s(t) - \C E_n(t)\big),
\end{split}\label{lem:upperBoundRR}
\end{align} 
where $W_l^s(t)$ is called  \textit{the data  coefficient} of stream $s$ over link $l$ and is defined as
\begin{align}
\begin{split}
W^s_l(t) = &   \big(U^s_{T(l)}(t)-\C E_{T(l)}(t)\big)  \iif{(T(l),s)\in \dbtilde{N}(t)} - \big(U^s_{R(l)}(t)-\C E_{R(l)}(t)\big)  \iif{(R(l),s)\in \dbtilde{N}(t)}  \\
& +U^s_{T(l)}(t)   \iif{(T(l),s)\in \tilde{N}(t)} -   U^s_{R(l)}(t) \iif{(R(l),s)\in \tilde{N}(t)}.
\end{split}
\label{eq:dataCoeff}
\end{align}
Furthermore,  $J_n(t)$ is called   \textit{the power coefficient} for node $n$ and is defined as 
\begin{align}
J_n(t) = \C\sum_s \big(U^s_{n}(t)-\C E_{n}(t)\big)  \iif{(n,s)\in \dbtilde{N}(t)}.
\label{eq:powerCoeff}
\end{align}

In order to minimize the right hand side of \eqref{lem:upperBoundRR}, it suffices to minimize the inner terms of the two expectations as the other terms are constant with respect to the control variables $\bm{p}(t), p_{AP}(t), \bm{w}(t)$ and $C_l^s(t)$. To minimize the first expectation, we first  allocate the whole capacity of each link to the stream with greatest data coefficient over that link, and then we select the minimizing power vector. Furthermore,  minimization of  the second expectation can be decomposed into beamforming vector selection and the EAP transmission power selection. We rewrite  the term inside the expectation as
\begin{align}
\begin{split}
 F(t) &\triangleq p_{AP}(t)\left[ V - \bm{w}(t)\left( \sum_{n=1}^L J_n(t) \bm{h}_{n}^T(t)\bm{h}_{n}(t)\right) \bm{w}^T(t)\right] \\
&=p_{AP}(t)\left[ V - \bm{w}(t)\bm{H}(t) \bm{w}^T(t)\right],
\end{split}
\label{eq:energyLinkControl}
\end{align}
where
\begin{align}
\bm{H}(t) =  \sum_{n = 1}^N J_n(t) \bm{h}_{n}^T(t)\bm{h}_{n}(t).
\label{eq:sumChan}
\end{align}
 It can be verified that the term inside the brackets is minimized if we select the beamforming vector  $\bm{w}^\ast(t)$ in direction of the eigenvector of $\bm{H}(t)$ with maximum eigenvalue. Substituting $\bm{w}^\ast(t)$ in \eqref{eq:energyLinkControl}, $F(t)$ is then minimized by determining  $p^\ast_{AP}(t)$  according to following rule

\begin{align}
p^\ast_{AP}(t) = 
\begin{cases}
\Pam, & V<  \sum_{n = 1}^N    |\bm{w}^\ast(t)\bm{h}^T_{n}(t)|^2J_n(t),  \\
0& \text{otherwise.}
\end{cases}
\end{align}

The data link control and energy link control polices are summarized in Algorithms \ref{alg:DataLink} and \ref{alg:EnergyLink}, respectively.
 
It should be noted that finding the optimal power vector $\bm{p}^\ast(t)$ in the data link control policy requires solving the max-weight problem 
$$ \bm{p}^\ast(t) = \argmin_{\bm{p}(t)\in\Pi}{   \sum_{l=1}^L  \big[ J_{T(l)}(t)  p_l(t)  -  W_l(t)C_l(\bm{p}(t),\bm{g}(t)) \big]},$$
which can be NP-hard in general. However,  in certain cases, e.g., in interference-free networks, closed-form solutions can be found. Furthermore, approximate solutions for this problem 
results in a bounded optimality gap in the overall performance.  The approximate solutions have been extensively discussed in \cite[Chapter~6]{Neely2010}.

\begin{algorithm}[t]
\caption{Data link control policy at each time-slot t (Data Routing)}\label{alg:DataLink}
\begin{algorithmic}[1]

\State Calculate $W_l^s(t),\;\forall l,s,$ and $J_l(t),\;\forall l$, according to \eqref{eq:dataCoeff} and \eqref{eq:powerCoeff}.
\State Find $W_l(t) \leftarrow \max_{s}\left \{W_l^s(t)\right\}$.
\State Find 
$\bm{p}^\ast(t) \leftarrow \argmin_{\bm{p}(t)\in\Pi}   \sum_{l=1}^L  \big[ J_{T(l)}(t)  p_l(t) -  W_l(t)C_l(\bm{p}(t),\bm{g}(t)) \big].$
\State Find $ s^\ast _l\leftarrow \argmax_{s\in[1,S]} {W_l^s(t)}$, and set $ C_l^{s^\ast_l}(t) \leftarrow C_l(\bm{p}^\ast(t),\bm{g}(t)),\; \forall l$.
\State Update data queues according to \eqref{eq:queueEvolve}.
\end{algorithmic}
\end{algorithm}

\begin{algorithm}[t]
\caption{Energy  link control policy at each time-slot t (Power Transfer Scheduling)}\label{alg:EnergyLink}
\begin{algorithmic}[1]
\State Calculate  $J_{n}(t),\;\forall n$, according to \eqref{eq:powerCoeff}.
\State Calculate the sum channel matrix  $\bm{H}$ according to \eqref{eq:sumChan}.
\State Derive the eigenvalues and eigenvectors of $\bm{H}$ and set $\bm{w}^\ast$ equal to the eigenvector with the largest eigenvalue.
\If{$V <  \sum_{n = 1}^N    |\bm{w}^\ast(t)\bm{h}^T_{n}(t)|^2J_n(t)$}
\State Set $p_{AP}(t) \leftarrow \Pam$.
\Else
\State Set $p_{AP}(t) \leftarrow 0$.
\EndIf

\end{algorithmic}
\end{algorithm}

\section{Performance Analysis of the Proposed Control Policies}\label{sec:performance}
 In this section, we first derive a lower bound on the minimum required power for stability. We then use Lyapunov Optimization Theorem \cite{Neely2010} to compare the proposed control policy  to the derived lower bound.

\subsection{Lower Bound on the Minimum Power for Stability}

In order to obtain a lower bound on the minimum power that stabilizes the queue of each link, we substitute the instantaneous battery constraint \eqref{eq:batteryConstraint} with a more relaxed constraint on the average power consumption. The battery constraint along with   \eqref{eq:batteryEvolve}   imply that
\begin{align}
\begin{split}
\sum_{t=0}^{T-1} \sum_{l\in\Out_n}p_l(t) \leq E_n(0) + \sum_{t=0}^{T-1} Q_n(t)\Rightarrow\limsup_{T\rightarrow \infty} {1\over T}\sum_{t=0}^{T-1} \sum_{l\in\Out_n} \eexp{p_l(t)} \leq \limsup_{T\rightarrow \infty} {1\over T}   \sum_{t=0}^{T-1} \eexp{Q_n(t)} \;\; \forall n, 
\end{split}
\label{eq:avgBatteryConstraint}
\end{align}
where a limited initial battery charge is assumed. The last inequality in \eqref{eq:avgBatteryConstraint} holds under any policy that conforms to the battery constraint. We name \eqref{eq:avgBatteryConstraint} the \textit{average power constraint} and use it as a substitute for the battery constraint \eqref{eq:batteryConstraint}.

Let  $\bm{\lambda} = (\lambda_1,\ldots, \lambda_S)$ denote the data streams arrival rate vector. We define the capacity region $\Lambda$ as the set of all data arrival rate vectors that can be stabilized under the average power constraint (interested readers are referred to \cite{Neely2003} for more details on the network capacity region). The following theorem introduces a randomized policy that achieves the minimum power consumption over all other polices with average power constraint.  

\color{black}{
\begin{theorem}
Suppose that channel states and data arrivals are i.i.d over different time-slots. Moreover, assume that the arrival rates belong to the capacity region (i.e., $\bm{\lambda} \in \Lambda$). The minimum power required for stability, $P^\ast_{AP}$, can be obtained by a stationary and probably randomized policy. This policy is a pure function of $\bm{H}(t)$, $\bm{g}(t)$ and $A_n^s(t),\; \forall n,s$,  with the following properties
\begin{align}
\lambda_n^s + \sum_{l \in \In_n}\eexp{ C_l^s(t)} &\leq \sum_{l \in \Out_n}\eexp{ C_l^s(t)}, \; \forall n,s,t,\\
\sum_{l \in \Out_n}\eexp{ p_l(t)} &\leq  \eexp{Q_n(t)},\;\forall n,t.
\end{align}
\begin{proof}
The proof   is similar to the proof of Theorem 4.5 in \cite{Neely2010}, and hence,  is omitted for brevity. 
\end{proof}
\label{th:randomOptimal}
\end{theorem}
}

Note that Theorem \ref{th:randomOptimal} only states that such stationary optimal policy with aforementioned properties exists, and does not derive such policy. In sequel, we use these properties to compare the average power consumption under our proposed policy to the  lower bound on the minimum required power for stability, i.e., $P^\ast_{AP}$.

\subsection{Performance of the Proposed Policy}
In this section, we derive the optimality gap of our proposed policy. Moreover, we show that the proposed policy stabilizes  the network and conforms to the battery constraint.
The following theorem summarizes the performance of the proposed policy,
\begin{theorem}\label{th:performanceTh}
Suppose the channel states and data arrivals are i.i.d over time-slots, and the arrival rates are strictly inside the capacity region, i.e., there is a scalar $\epsilon_{max}$ such that $\forall \epsilon \in (0,\epsilon_{max}]: \bm{\lambda} + \bm{\epsilon} \in \Lambda$, where $\bm{\epsilon}$ is a vector with all entries equal to $\epsilon$ . Under the proposed policy,
\begin{enumerate}
\item At any time-slot $t$, the transmission power assigned to data links originated from node $n$ are nonzero only if its battery level is higher than the maximum data transmission power, i.e., $E_n(t) > \Pm$.
\item The time average expected power consumption  satisfies,
\begin{align}
\limsup_{T\rightarrow \infty}{1\over T} \sum_{t = 0}^{T-1}\mathbb{E}\{ p_{AP}(t)\} &\leq p^{\ast}_{AP} + {B \over V}.
\label{eq:optimalityGap}
\end{align}
\item The queues are stable and time average expected sum backlog satisfies,
\begin{align}
\limsup_{T\rightarrow \infty}{1\over T} \sum_{t = 0}^{T-1} \sum_{n,s}\mathbb{E}\{U^s_n(t)\} & \leq {V p^{\ast}_{AP} + B^\prime \over \epsilon_{max}},
\label{eq:stability}
\end{align}
where $B^\prime = B + \epsilon_{max}NSU_0$.
\end{enumerate}
\end{theorem}

It should be noted that part 1 in Theorem \ref{th:performanceTh} guarantees that our proposed policy does not violate the battery level constraint. Moreover, parts 2 and 3 show the optimality of the power consumption  and stability of the network under our proposed policy, respectively. 
 \begin{remark}
Note that the performance bounds in $\eqref{eq:optimalityGap}$ and $\eqref{eq:stability}$ introduce a trade-off between the optimality gap and the average queue backlog.  According to this trade-off, when the average power consumption is within  $\mathcal{O}({1\over V})$ of the minimum required power, the  average backlog  could be upper bounded by a term of the order of $\mathcal{O}(V)$. 
\end{remark}
Now we prove Theorem \ref{th:performanceTh}.

\begin{proof}
 Part 1 is proven in Appendix \ref{sec:ProofBatteryCons}, optimality (part 2)  and stability (part 3) are proven here.  
 Suppose that the arrival rate is $\bm{\lambda}+ \bm{\epsilon}$. Since the arrivals are i.i.d,  according to Theorem \ref{th:randomOptimal} there is a stationary randomized policy with the following properties, 
\begin{align}
\begin{split}
\lambda_n^s+\epsilon + \eexp{\sum_{l\in\In_n} C_l^s(t)} &\leq \eexp{\sum_{l\in\Out_n} C_l^s(t)} \;\; \forall n,s,\\
\eexp{ \sum_{l \in \Out_n}  p_l(t)   }& \leq \eexp{ Q_n(t)  } \;\; \forall n,\\
\eexp{p_{AP}(t)} &= p^{\ast}_{AP}(\epsilon),
\end{split}
\label{eq:rndPolicyProp}
\end{align}
where $p^{\ast}_{AP}(\epsilon)$ is the minimum power required for stability when the arrival rate equals $\bm{\lambda} + \bm{\epsilon}$. Let $\mathcal{P}^{S}_{\epsilon}$ denote the above stationary policy. Our proposed policy minimizes the right hand side of \eqref{eq:dppUpperBound} over any alternative policy including  $\mathcal{P}^{S}_{\epsilon}$.  Plugging the properties in \eqref{eq:rndPolicyProp} into right hand side of \eqref{eq:dppUpperBound} yields, 
\begin{align*}
\Delta(L(t)) &+V \eexp{   p_{AP}(t)| \bm{U}(t), \bm{E}(t)}  \leq B+Vp_{AP}^{\ast}(\epsilon) - \epsilon \sum_{(n,s)\in \tilde{N}(t)}  U_n^s(t).
\end{align*}
Taking expectation with respect to $ \bm{U}(t)$ and $\bm{E}(t)$ from both sides results in
\begin{align*}
\eexp{L(t+1)}& - \eexp{L(t)}  + \eexp{p_{AP}(t)}  \leq B+Vp_{AP}^{\ast}(\epsilon) - \epsilon \sum_{(n,s)\in \tilde{N}(t)}  U_n^s(t).
\end{align*}
Then, summing both sides over $t = 0,\ldots, T-1$ yields

\begin{align*}
\eexp{L(T-1)} - \eexp{L(0)} + \sum_{t=0}^{T-1}\eexp{p_{AP}(t)} \leq \big(B+Vp_{AP}^{\ast}(\epsilon)\big)T - \epsilon \sum_{t=0}^{T-1} \sum_{(n,s) \in \tilde{N}(t)} U_n^s(t).
\end{align*}
Now by rearranging the terms and dropping the negative terms when appropriate, we get the following inequalities:
\begin{align}
{1 \over T}\sum_{t=0}^{T-1}\eexp{p_{AP}(t)} &\leq {B\over V} +p_{AP}^{\ast}(\epsilon) + {\eexp{L(0)}\over T},\label{eq:opt1}\\
{1 \over T} \sum_{t=0}^{T-1} \sum_{(n,s) \in \tilde{N}(t)} U_n^s(t) &\leq {B+Vp_{AP}^{\ast}(\epsilon) \over \epsilon} +  {\eexp{L(0)}\over T},\label{eq:stab1}
\end{align}
The bounds in \eqref{eq:opt1} and \eqref{eq:stab1} can be separately optimized over values of $ \epsilon \in (0, \epsilon_{max}] $. Since $\lim_{\epsilon \rightarrow 0}p_{AP}^{\ast}(\epsilon) =0$, letting $\epsilon \rightarrow 0$ in \eqref{eq:opt1} and taking limits as $T\rightarrow \infty$ concludes the second statement of Theorem \ref{th:performanceTh}. To prove the last statement, we first  substitute the inner summation in \eqref{eq:stab1} with summation over all data queues and then add   the term $NSU_0$ to the right hand side as a compensation.  Setting $\epsilon = \epsilon_{max}$ and taking limits as $T\rightarrow \infty$ completes the proof of the third statement  of Theorem \ref{th:performanceTh}.
\end{proof}

\section{Simulation Results}\label{sec:sim}

\begin{figure}
\centering
\begin{tikzpicture}
   [WD/.style={circle,draw,fill=none,thick,scale =1},
   HAP/.style={rectangle ,draw,fill=none,thick,scale =2},
     ST/.style={rectangle ,fill=none,thick,scale =1},scale=2]

     \node(HAP) at ( 0,0) [HAP,label=below: EAP]{};
     \node(WD1) at ( -1.4,1) [WD,label=above:N1 (source 1)] {};
     \node(WD2) at (-1,0)  [WD,label=below:  N2 (source 2)] {};
     \node(WD3) at (0,0.7)  [WD,label=above: N3] {};
     \node(WD4) at (1,1)  [WD,label=above: N4 (sink 1)] {};
     \node(WD5) at (1,0)  [WD,label=below: N5 (sink 2)] {};
     \node(ST1) at (-2,1)  [ST ] {};
     \node(ST2) at (-1.6,0)  [ST ] {};
 	
     \draw [->,  shorten >=0.1cm,shorten <=0.1cm , line width = 0.3mm] (WD1.east) to (WD3.west); 					%WD1 to WD3
     \draw [->,  shorten >=0.1cm,shorten <=0.1cm ,line width = 0.3mm] (WD2.north east) to (WD3.west); 	%WD2 to WD3
     \draw [->,  shorten >=0.1cm,shorten <=0.1cm ,line width = 0.3mm] (WD3.east) to (WD4.west); 					%WD3 to WD4
     \draw [->,  shorten >=0.1cm,shorten <=0.1cm ,line width = 0.3mm] (WD3.east) to (WD5.west); 		%WD3 to WD5
     \draw [<->,  shorten >=0.1cm,shorten <=0.1cm ,line width = 0.3mm] (WD4.south) to (WD5.north); 				%WD3 to WD5
     \draw [->,  shorten >=0.1cm,shorten <=0.1cm ,dashed,line width = 0.3mm, color = red] (HAP.west) to (WD2.east); 					 % HAP to WD1		
     \draw [->,  shorten >=0.1cm,shorten <=0.1cm ,dashed,line width = 0.3mm, color = red] (HAP.north west) to (WD1.south east); 		 % HAP to WD2		
     \draw [->,  shorten >=0.1cm,shorten <=0.1cm ,dashed,line width = 0.3mm, color = red] (HAP.north) to (WD3.south); 			        % HAP to WD3
     \draw [->,  shorten >=0.1cm,shorten <=0.1cm ,dashed,line width = 0.3mm, color = red] (HAP.north east) to (WD4.south west); 		 % HAP to WD4
     \draw [->,  shorten >=0.1cm,shorten <=0.1cm ,dashed,line width = 0.3mm, color = red] (HAP.east) to (WD5.west); 			        	 % HAP to WD4
     \draw [->, dashed,line width = 0.5mm] (ST1.east) to (WD1.west); 			        % Stream1 
     \draw [->, dashed,line width = 0.5mm] (ST2.east) to (WD2.west); 			        % Stream 2 
\end{tikzpicture} 
\caption{The considered topology. The black and  dashed red arrows represent the data links and energy links, respectively.}
\label{fig:simTopol}
\end{figure}
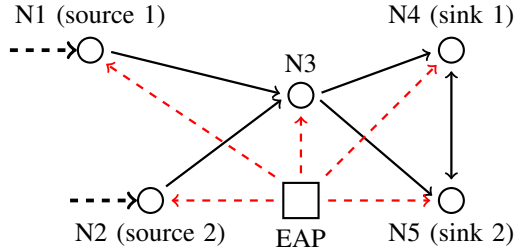

In this section, we consider a wireless network consisting of one EAP and five wireless nodes, as shown in Fig.  \ref{fig:simTopol}. There are two streams of data, from node 1 to node 4 and from node 2 to node 5, with average arrival rates of $\lambda_1 = 100$ bit/slot and $\lambda_2 = 50$ bit/slot, respectively. The data links and energy links channel states are generated according to Rician fading model \cite{Rician2015}, with Rician factor equal to 1. The EAP is equipped with $M=50$ antennas that are configured as a half wavelength separated array. Moreover, similar to existing works in literature (e.g see \cite{Gurakan2016}    ) we assume no interference across the data links and consider an AWGN model for their capacity, as follows
\begin{align*}
C_l(\bm{p}(t),\bm{g}(t)) = W\log\left(1+ {p_l(t)|g_l(t)|^2\over WN_0}\right),
\end{align*}
where $W = 10$ kHz  and $N_0 = -135$ dBm/Hz are  the channel bandwidth and the noise spectral density, respectively. Finally, the maximum transmission power of the EAP and the nodes are considered to be $\Pm = \SI{4}{\micro\watt}$ and $\Pam = \SI{4}{\watt}$, respectively.  All numerical results have been obtained by running the simulation for $6\times 10^6$ time-slots using Matlab 2015a on a simulation platform with 20 cores and 256 GB of  RAM.

Fig. \ref{fig:res} shows the average power consumption of the EAP as well as the average backlog of the data queues in the network, versus the trade off parameter $V$. As can be seen in Figures \ref{fig:avgPower} and \ref{fig:avgBacklog}, the average power consumption decays very fast as $V$ increases, while the average data queue backlog increases linearly with $V$.   Such behavior  complies with our theoretical results derived in  \eqref{eq:optimalityGap} and   \eqref{eq:stability}.
\begin{figure}
\centering
\begin{subfigure}[b]{0.4\textwidth}
\includegraphics[width =\textwidth]{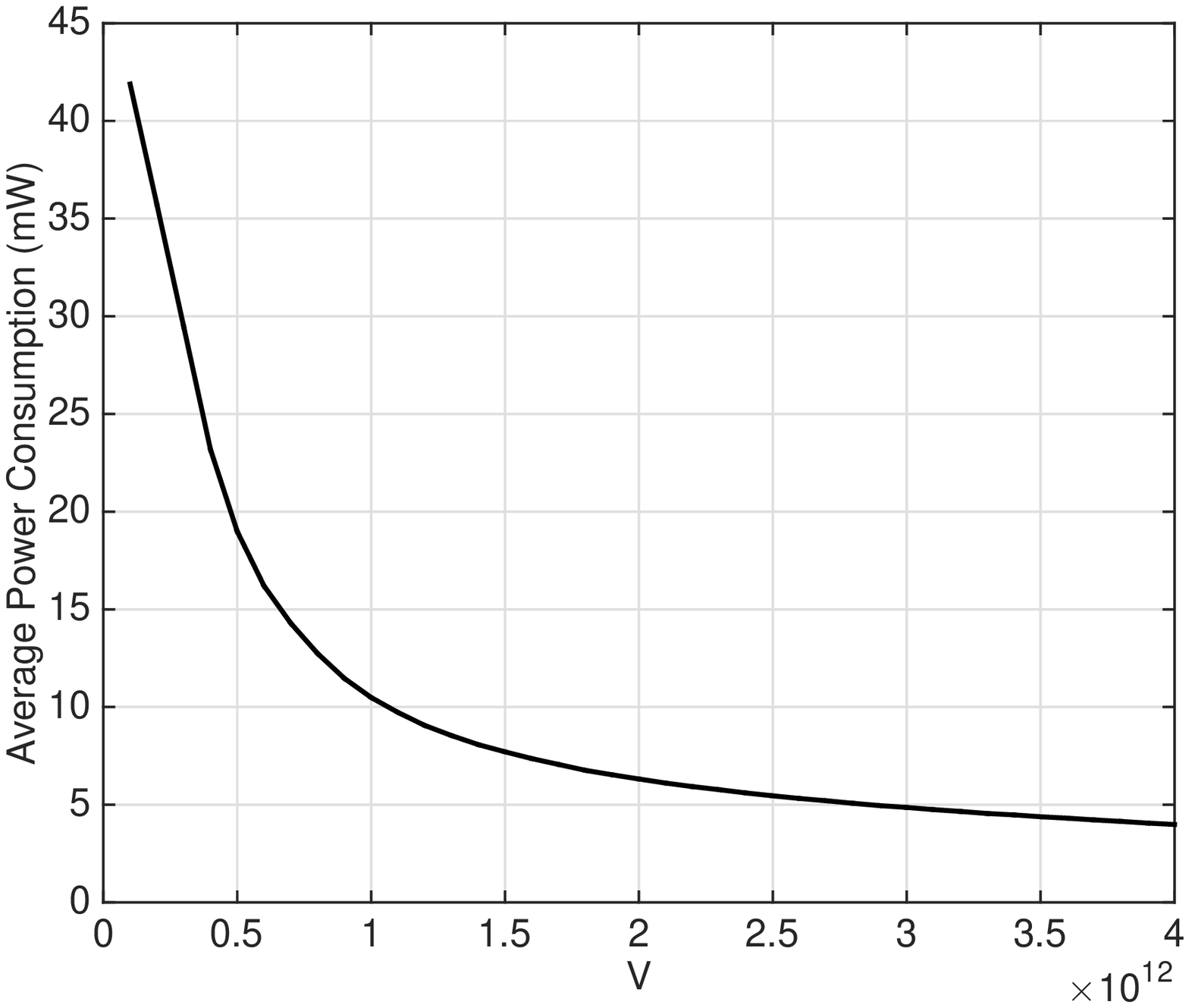}
\caption{Average power consumption}
\label{fig:avgPower}
\end{subfigure}
~
\begin{subfigure}[b]{0.4\textwidth}
\includegraphics[width = \textwidth]{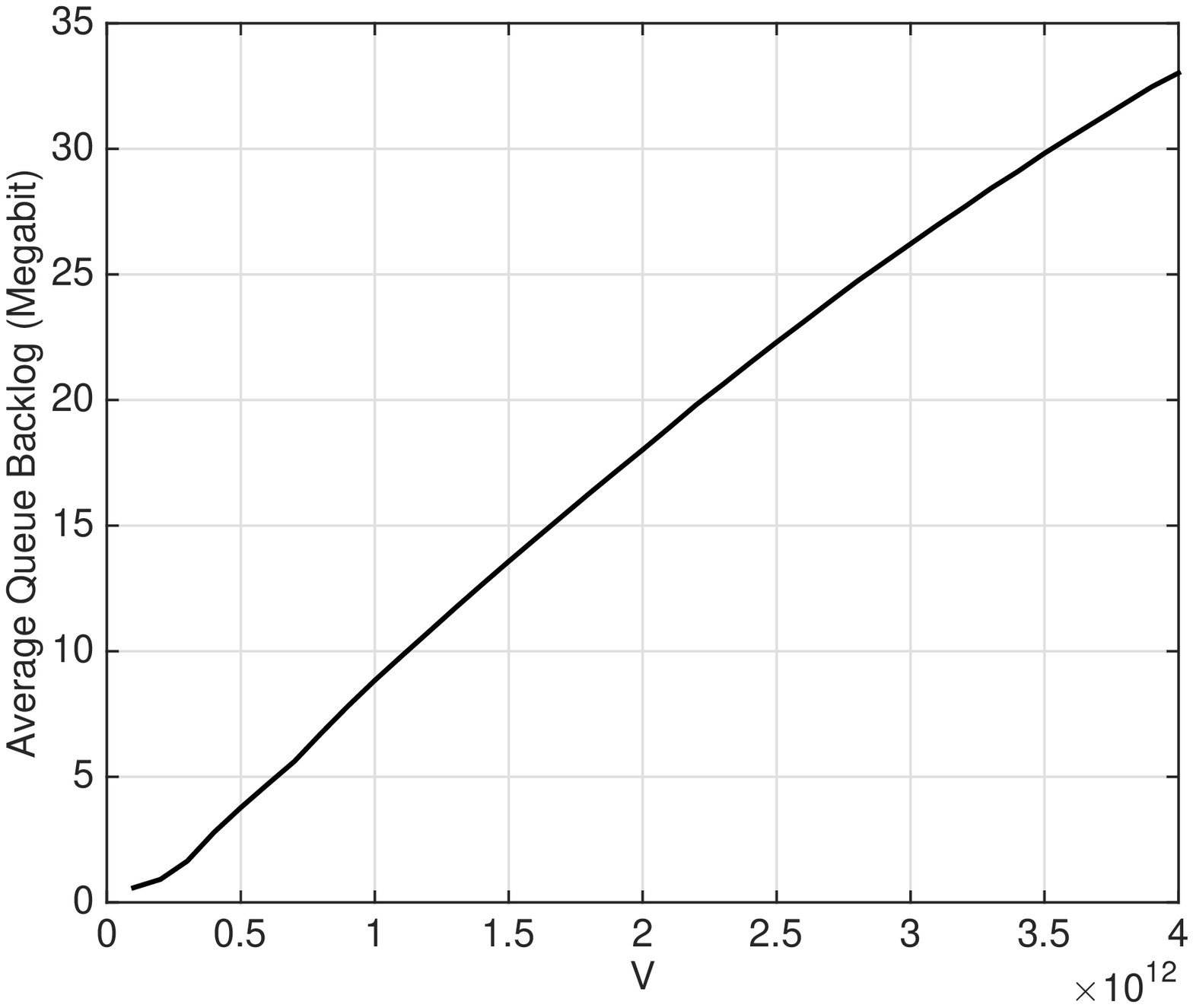}
\caption{Average queue backlog}
\label{fig:avgBacklog}
\end{subfigure}
\caption{Average EAP power consumption and queue backlog versus $V$.}
\label{fig:res}
\end{figure}

Next, Fig. \ref{fig:sampPath} depicts a sample path for the data queue backlog process and the  battery level process of node 1 (for $V = 2\times 10^{12}$).  It can be verified from this figure  that the queue backlog is stabilized around $75$ Megabits while no energy outage has occurred. 

\begin{figure}
\centering
\begin{subfigure}[b]{0.4\textwidth}
\includegraphics[width = \textwidth]{SampleData.eps}
\caption{Data queue backlog}
\label{fig:sampBackLog}
\end{subfigure}
~
\begin{subfigure}[b]{0.4\textwidth}
\includegraphics[width =\textwidth]{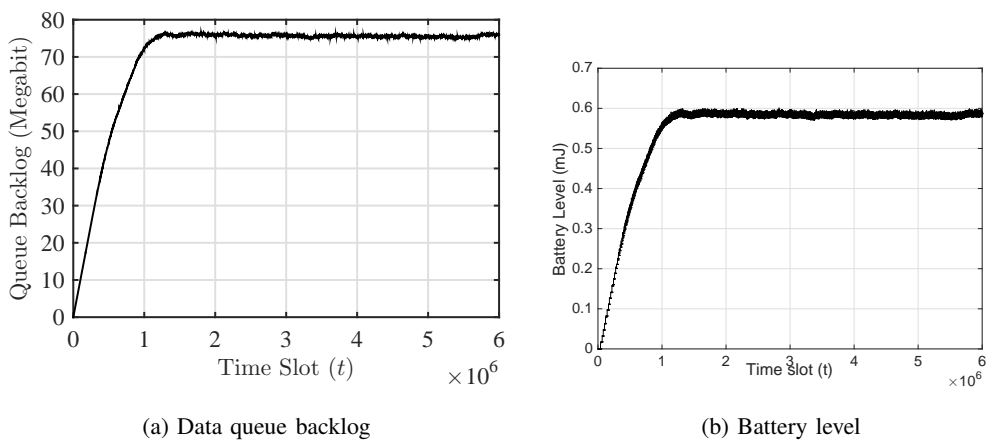}
\caption{Battery level}
\label{fig:sampBattery}
\end{subfigure}
\caption{A sample path for data queue backlog and battery process.}
\label{fig:sampPath}
\end{figure}
Finally, Fig. \ref{fig:AvgPat} shows the average transmission pattern of the EAP  (for $V = 2\times 10^{12}$). As can be  clearly  seen in this figure, there are three distinguished peaks in the transmission patter of the EAP at the direction toward  node 1, 2 and 3 which are the most congested nodes in this network. Note that the peak at $0^\circ$ is due to the  linear structure of the EAP antenna (the gains at $0^\circ$ and $180^\circ$ for a linear array are reciprocal). Moreover, as can be clearly verified from this figure, the maximum value of the pattern is at the direction toward node 3. This is due to the fact that all network traffic pass through this node.
\begin{figure}
\centering
\includegraphics[width = 0.4\textwidth]{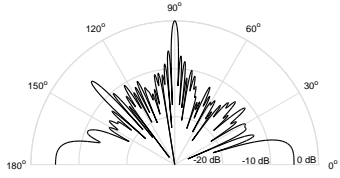}
\caption{Average transmission pattern of the EAP. }
\label{fig:AvgPat}
\end{figure}

%
%\begin{figure}
%\centering
%\includegraphics[width = 0.45\textwidth]{images/AveragePower.eps}
%\caption{Average Power Consumption}
%\label{fig:avgPower}
%\end{figure}
\section{Conclusion}\label{sec:conclude}
In this paper, we focused on a wireless powered communication network with battery-operated nodes and proposed  a joint power allocation, data routing and energy beamforming policy to minimize the average power consumption in the network. The proposed policy adapts to general networks with arbitrary channel models, without any knowledge of the channel statistics. By theoretical analysis, we proved that our proposed policy conforms to the battery constraint and stabilizes the network. Moreover, we derived the optimality gap for the average power consumption under this policy. Finally, various numerical results are provided to show the significant performance of the proposed solution.

\appendices
\section{Upper Bound for Drift Plus Penalty Function }\label{sec:proofUpperBound}
Here we prove the inequality in  \eqref{eq:dppUpperBound} holds. We enumerate three different cases for $U^s_n(t)$ and $E_n(t)$, and bound the increment of $L^s_n(t)$ in successive time-slots for each case: 

\begin{itemize}

\item $U^s_n(t) \leq \Ut$:

According to \eqref{eq:LyapunovFun} the increment of $L^s_n(t)$ can be bounded as,
\begin{align}
L^s_n(t+1)-&L^s_n(t) \leq L^s_n(t+1)  \leq \oot U^s_n(t+1)^2 + \oot U^s_n(t+1)^2=  U^s_n(t+1)^2.
\end{align}
From  $U^s_n(t) \leq \Ut$ and \eqref{eq:queueEvolve} we know $U^s_n(t+1) \leq \Ut + \Am+\Ni\Cm$,  where $\Ni$ is the number of ingoing links to the node with most ingoing links in the network,  using this inequality in the above we get, 
\begin{align}
L^s_n(t+1)-L^s_n(t) \leq (\Ut + \Am+\Ni\Cm)^2 \triangleq B_0
\end{align}

\item $U^s_n(t) > \Ut$ and $U^s_n(t) \leq \C E_n(t)$: 

In this case we have, 
\begin{align*}
L^s_n(t+1)-L_n^s(t) &= \oot \big ( U^s_n(t+1)^2 -  U^s_n(t)^2 \big ) + \oot \big(U^s_n(t+1)-\C E_n(t+1)\big)^2 \mathbf{1}_{\left(U^s_n(t+1) > \C E_n(t+1)\right)}\\
&\leq \oot \left ( U^s_n(t+1)^2 -  U^s_n(t)^2 \right ) \oot \big(U^s_n(t+1)-\C E_n(t+1)\big)^2\\
&\overset{a}{\leq}  \left [A^s_n(t) +\sum_{l\in \In_n} C_l^s(t) - \sum_{l\in \Out_n} C_l^s(t) \right] U^s_n(t) + \oot(\No^2+\Ni^2)\Cm^2+\oot \Am^2+\Am\Ni\Cm \\
&\qquad+\oot \big(U^s_n(t+1)-\C E_n(t+1)\big)^2, \label{eq:inequality0}
\end{align*}
  where $\No$ is the number of outgoing links from the node with most outgoing links in the network. Inequality $a$ is achieved by substituting $U^s_n(t+1)$ with the LHS of \eqref{eq:queueEvolve}. 
Finally from $U^s_n(t) \leq \C E_n(t)$, \eqref{eq:queueEvolve}  and \eqref{eq:batteryEvolve} it is easy to verify that,
\begin{align*}
 \oot \big(U^s_n(t+1)-&\C E_n(t+1)\big)^2 \leq\oot \big(\Am +\Ni\Cm+\C\No\Pm\big)^2, 
\end{align*}
and we conclude,
\begin{align}
&L^s_n(t+1)-L^s_n(t) \leq \left [A^s_n(t) +\sum_{l\in \In_n} C_l^s(t) - \sum_{l\in \Out_n} C_l^s(t) \right] U^s_n(t)  + B_1,
\end{align}
where $B_1 \triangleq \oot(\No^2+\Ni^2)\Cm^2+\oot \Am^2+\Am\Ni\Cm +\oot \big(\Am +\Ni\Cm+\C\No\Pm\big)^2 $.

\item $U^s_n(t) > \Ut$ and $U^s_n(t) > \C E_n(t)$:

For this case we have:
\begin{align*}
L^s_n(t+1)-L^s_n(t) \leq&\oot U^s_n(t+1)^2 - \oot U^s_n(t)^2  \\
&+ \oot \left(U^s_n(t+1)-\C E_n(t+1)\right)^2 \\
&-\oot\left(U^s_n(t)-\C E_n(t)\right)^2.
\end{align*}
Replacing $U^s_n(t+1)$ and $E_n(t+1)$ from \eqref{eq:queueEvolve} and \eqref{eq:batteryEvolve} in the above, after some algebraic manipulation we get,
\begin{align*}
L^s_n(t+1)-L^s_n(t) &\leq \left [A^s_n(t) +\sum_{l\in \In_n} C_l^s(t) - \sum_{l\in \Out_n} C_l^s(t) \right] U^s_n(t)  \\
 &  +\left [A^s_n(t) +\sum_{l\in \In_n} C_l^s(t) - \sum_{l\in \Out_n} C_l^s(t) \right] \big(U^s_n(t) - \C E_n(t)\big) \\
&  - \C\left[Q_n(t)-\sum_{l\in \Out_n}p_l(t)  \right]\big(U^s_n(t) - \C E_n(t)\big)+B_2,
\end{align*}
where $B_2 \triangleq (\No^2+\Ni^2)\Cm^2 + \Am^2+2\Am\Ni\Cm+\oot\C^2\No^2\Pm^2+\oot\C^2\Pam^2+\C\Pam+\oot\No\Cm +\No\Ni\Pm\Cm+\Am\No\Pm $
\end{itemize}

Considering the above three cases and taking summation over $L^s_n(t+1) - L^s_n(t)$ for $n = 1,\ldots,N$ and $s = 1, \ldots, S$ we would have,
\begin{align}
\begin{split}
\Delta L(t)  = \sum_{n=1}^N \sum_{s =1}^S L^s_n(t+1) - L^s_n(t) &\leq B+\sum_{(n,s)\in \tilde{N}(t)}  \left [A^s_n(t) +\sum_{l\in \In_n} C_l^s(t) - \sum_{l\in \Out_n} C_l^s(t) \right] U^s_n(t) \\
&+ \sum_{(n,s)\in \dbtilde{N}(t)} \bigg [A^s_n(t) +\sum_{l\in \In_n} C_l^s(t) - \sum_{l\in \Out_n} C_l^s(t) \bigg] \big(U^s_n(t) - \C E_n(t)\big)\\
&- \C\sum_{(n,s)\in \dbtilde{N}(t)}\left[Q_n(t)-\sum_{l\in \Out_n}p_l(t)  \right]\big(U^s_n(t) - \C E_n(t)\big),
\end{split}\label{eq:inequalityDrift}
\end{align}
where $B \triangleq N\max\{B_0,B_1,B_2\}$. Adding $VP_A(t)$ to both sides of \eqref{eq:inequalityDrift}, taking expectation conditioned on $\bm{U}(t)$ and $\bm{E}(t)$ and rearranging the terms  proves the intended  result.

\section{The Proposed Policy Conforms to Battery Constraint}\label{sec:ProofBatteryCons}
Here we prove part 1 of Theorem \ref{th:performanceTh}.  Let us assume $ E_n(t) <\Pm$ for a specific node $n$. Consider data link $l$ such that $T(l) =n$ and a power vector $\bm{p}(t)$. Let  us define another power vector, $\bm{p}_0(t)$, by setting the $l$th entry in $\bm{p}(t)$ to zeros. The  transmission power for data links are determined by the solution of the  minimization problem,
\begin{align}
\argmin_{\bm{p}(t)\in\Pi}{ G(\bm{p}(t)) =   \sum_{l=1}^L  \big[ J_{T(l)}(t)  p_l(t)  -  W_l(t)C_l(\bm{p}(t),\bm{g}(t)) \big]}.
\end{align}
%
% in Eq.  \eqref{prob:dataLinkPower}. Let us denote by $G(\bm{p}(t))$ the objective function in \eqref{prob:dataLinkPower}, i.e.,
%\begin{align*}
%G(\bm{p}(t)) = \sum_{l=1}^L  \big[ J_l(t)  p_l(t)  -  W_l(t)C_l(\bm{p}(t),\bm{g}(t)) \big].
%\end{align*}
To prove the intended result, it suffices to show $G(\bm{p}(t)) - G(\bm{p}_0(t)) \geq0$. The following always hold,
\begin{align}
G(\bm{p}(t))  G(\bm{p}_0(t)) &= p_l(t)J_{T(l)}(t) - W_l(t) \big[ C_l(\bm{p}(t),\bm{g}(t))  - C_l(\bm{p}_0(t),\bm{g}(t))\big]
\\
& -\;\;\sum_{l^\prime\neq l} W_{l^\prime}(t) \big[C_{l^\prime}(\bm{p}(t),\bm{g}(t) )- C_{l^\prime}(\bm{p}_0(t),\bm{g}(t))\big]\\ 
&\geq p_l(t)J_{T(l)}(t) - W_l(t) \big[ C_l(\bm{p}(t),\bm{g}(t)) - C_l(\bm{p}_0(t),\bm{g}(t))\big]\label{eq:ineqa}\\
&\geq p_l(t)J_{T(l)}(t) - \delta p_l(t)W_l(t),\label{eq:ineqb}
\end{align}
where \eqref{eq:ineqa} and \eqref{eq:ineqb} are due to the properties of the capacity function in  \eqref{eq:interference} and  \eqref{eq:rateBound}. It can be verified that $E_n(t) < \Pm$ together with  $U_n^s \in \tilde{N}(t)$ contradicts  $U_n^s \notin \dbtilde{N}(t)$. Plugging this property into definition of $W_l(t)$ and neglecting negative terms yields,
\begin{align*}
W_l(t) < \sum_{s: (T(l),s)\in \dbtilde{N}(t)} [ 2U_{T(l)}^s(t) - \C E_{T(l)}(t) ]. 
\end{align*}
Using this inequality in \eqref{eq:ineqb}, we get (all summations are over $\{ s:(T(l),s)\in \dbtilde{N}(t)\}$), 
\begin{align}
\begin{split}
G(\bm{p}(t)) - G(\bm{p}_0(t)) &\geq \C p_l(t) \sum_{s}\;\; [ U_{T(l)}^s(t) - \C E_{T(l)}(t)  - \delta p_l(t) \sum_{s} \;\;[ 2U_{T(l)}^s(t) - \C E_{T(l)}(t) ]\end{split}\\
\begin{split}
& = p_l(t) \sum_{s}\;\; \left[ U_{T(l)}^s(t) - \C E_{T(l)}(t) \right]\bigg[\C - \delta \left(1+\frac{ \sum_{s}\;\; U_{T(l)}^s(t)   }{ \sum_{s}\;\; [ U_{T(l)}^s(t) - \C E_{T(l)}(t) ] }\right) \bigg ]\end{split} \\
\begin{split}
& \geq p_l(t) \sum_{s}\;\; \left[ U_{T(l)}^s(t) - \C E_{T(l)}(t) \right]\bigg [\C - \delta \left(1+ \frac{ U_0 }{  U_0 - \C \Pm }\right) \bigg ]\label{eq:batProofE2}\end{split}\\
\begin{split}
& = p_l(t) \sum_{s}\;\; \left[ U_{T(l)}^s(t) - \C E_{T(l)}(t) \right] \left(\C\left (1-\frac{1}{\alpha}\right)-2\delta\right)\\
&= 0.\label{eq:batProofE3}
\end{split}
\end{align}

Inequality \eqref{eq:batProofE2} is due to the fact that $U_{T(l)}^s > U_0$ and $E_{T(l)} < \Pm$. The  equality  in \eqref{eq:batProofE3} can be verified by plugging $U_0 = CP_m+\alpha\delta\Pm$ in \eqref{eq:batProofE2}. Finally, noting that $\C > \frac{2\delta}{\left (1-\frac{1}{\alpha}\right)}$ the intended result is proved.
\bibliographystyle{IEEEtran}
\bibliography{paper}

\end{document}